\documentclass[onecolumn,11pt,aps,nofootinbib,superscriptaddress,tightenlines,showkeys]{revtex4}

\usepackage[T1]{fontenc}
\usepackage[utf8]{inputenc}

\usepackage[cmex10]{amsmath}
\usepackage{graphics}
\usepackage{dsfont}
\usepackage{amssymb}
\usepackage{graphicx}
\usepackage{mathrsfs}
\usepackage{units}
\usepackage{pgfplots}
\usepackage{enumitem}
\usepackage[colorlinks=true,linkcolor=black,citecolor=black,plainpages=false,pdfpagelabels]{hyperref}
\hypersetup{pdftitle={Eigenvalue estimates for the resolvent of a non-normal matrix}}

\usepackage{amsthm}
\newtheorem{theorem}{Theorem}[section]

\newtheorem{proposition}[theorem]{Proposition}
\newtheorem{corollary}[theorem]{Corollary}
\newtheorem{definition}{Definition}[section]

\newenvironment{remark}[1][Remark]{\begin{trivlist}
\item[\hskip \labelsep {\bfseries #1}]}{\end{trivlist}}

\newcommand{\id}{\ensuremath{\mathds{1}}}

\begin{document}
\title{Dynamical Quantum Tomography}
\author{Michael Kech}
\email{kech@ma.tum.de}
\affiliation{Department of Mathematics, Technische Universit\"{a}t M\"{u}nchen, 85748 Garching, Germany}

\date{\today}

\begin{abstract}
We consider quantum state tomography with measurement procedures of the following type: First, we subject the quantum state we aim to identify to a know time evolution for a desired period of time. Afterwards we perform a measurement with a fixed measurement set-up. This procedure can then be repeated for other periods of time, the measurement set-up however remains unaltered. Given an $n$-dimensional system with suitable unitary dynamics, we show that any two states can be discriminated by performing a measurement with a set-up that has $n$ outcomes at $n+1$ points in time. Furthermore, we consider scenarios where prior information restricts the set of states to a subset of lower dimensionality. Given an $n$-dimensional system with suitable unitary dynamics and a semi-algebraic subset $\mathcal{R}$ of its state space, we show that any two states of the subset can be discriminated by performing a measurement with a set-up that has $n$ outcomes at $l$ steps of the time evolution if $(n-1)l\ge 2\dim\mathcal{R}$. In addition, by going beyond unitary dynamics, we show that one can in fact reduce to a set-up with the minimal number of two outcomes.

\end{abstract}

\maketitle


%
%
%
\section{Introduction and Summary}
Quantum tomography is the task of identifying an unknown quantum state form the outcomes of a measurement. It is an integral part of quantum information science its implementation, however, is expensive. Yet, in some relevant scenarios it can be simplified: If prior information constrains the set of states to a subset of lower dimensionality the number of measurement outcomes necessary to uniquely identify a state can reduce considerably. In particular pure state tomography, or more generally tomography on states of bounded rank, has received significant attention and still is a field of active research.

Methods to find lower bounds on the number of measurement outcomes necessary to discriminate any two states of a given subset of the state space were first provided in \cite{heinosaari2013quantum} and later it was shown in \cite{kech} that these method apply in a rather general framework. However, from a practical point of view not all measurements might be feasible for implementation and thus one might want to restrict 
to a set of admissible measurements. Doing so, it is not clear whether the lower bounds established in \cite{heinosaari2013quantum,kech} still apply. In the context of pure state tomography it was shown in \cite{mondragon2013determination,jaming2014uniqueness,carmeli2015many,kech2} that any two pure states can be discriminated by performing four von-Neumann measurements, which is indeed tight for systems of dimension $n\geq5$. Additionally,  in \cite{kech2} this result was extended to more general subsets of the state space, including states of bounded rank. In the closely related fields of phase retrieval and low rank matrix recovery similar questions were addressed in \cite{balan2006signal,conca2014algebraic,edidin2015projections,vinzant2015small,kueng2014low,xu2015minimal}.  Finally, at the cost of requiring slightly more measurement outcomes, robust reconstruction algorithms are provided in \cite{QCS2,QCS3,GrossPhaseLift}.

\textit{Purpose of the present paper.}
The conventional approach to quantum tomography is to design a certain measurement set-up. Performing a statistical experiment, the state is identified from the relative frequencies of the measurement outcomes. If the system is of dimension $n$, at least $n^2$ outcomes are required to identify an unknown state. In the present paper we consider a more general scenario. Suppose we are given a measurement set-up and that, in addition, the system can be evolved according to a know time evolution. Rather than performing a conventional measurement with the set-up, we take advantage of the time evolution by considering measurement procedures of the following kind: Having evolved the system for a desired period of time, we perform a measurement with the given measurement set-up. Then, the procedure can be repeated for other periods of time.  Using this measurement scheme, we show that for suitable time evolutions any state can be identified with a measurement set-up that has solely two outcomes. Furthermore, also considering the scenario where prior information constrains the relevant set of states to a subset of lower dimensionality, we provide upper bounds on the minimal number of points in time on which one has to perform a measurement in order to be able to discriminate any to states of a given subset of the state space.


In the present paper we do not consider the algorithmic problem of reconstructing the state from the measurement data.

\textit{Outline.}
In Section II, we fix notation and introduce notions that are relevant for the following.

In Section III, we consider systems with discrete unitary dynamics. The first part is devoted to informationally complete tomography. Given the possibility to perform a measurement with a given measurement set-up at several time steps of the unitary evolution, we show that this set-up has to have at least $n$ outcomes to perform an informationally complete tomography if the system is $n$-dimensional. Furthermore, we show that under some condition on the time evolution, an informationally complete tomography can be performed by measuring with a set-up that has $n+1$ outcomes at $n$ time steps. In the second part we consider tomography on subsets of the state space. We show that performing a measurement with a set-up that has $m\ge n$ outcomes at sufficiently many points in time is a universal measurement scheme in the sense of \cite{kech2}. This allows us to prove a Whitney type embedding result: Given an $n$-dimensional system with suitable unitary dynamics and a semi-algebraic subset $\mathcal{R}$ of its state space, we show that any two states of the subset can be discriminated by performing a measurement with a set-up that has $m\geq n$ outcomes at $l\ge \frac{2\dim\mathcal{R}}{m-1}$ steps of the time evolution. Furthermore, we show that any two states of an $n$-dimensional system whose rank is at most $r$ can be discriminated by performing a measurement with a set-up that has $m\ge n$ outcomes at $l\ge\frac{4r(n-r)-1}{m-1}$ time steps. This upper bound on the number of time steps is close to the lower bound established in \cite{heinosaari2013quantum,kech}.

In Section IV, we generalize the system dynamics to a larger class of discrete CPTP time evolutions. Just like in Section III we prove a universality result in the sense of \cite{kech2}. Different form the case of unitary system dynamics, in this case there is just the trivial lower bound on the number of outcomes of the measurement set-up and indeed an informationally complete tomography can be performed by measuring with a set-up that has just two outcomes at $n^2-1$ points in time. Similar to the last section, given an $n$-dimensional system with suitable CPTP dynamics and a semi-algebraic subset $\mathcal{R}$ of its state space, we show that any two states of the subset can be discriminated by measuring with a set-up that has $m\geq 2$ outcomes at $l\ge \frac{2\dim\mathcal{R}}{m-1}$ steps of the time evolution. Furthermore, we show that any two state of an $n$-dimensional system of rank at most $r$ can be discriminated by measuring with a measurement set-up that has $m\geq 2$ outcomes at $l\ge\frac{4r(n-r)-1}{m-1}$ steps of the time evolution.

Having solely dealt with discrete time evolutions before, in appendix B we consider the possibility of performing measurements at rational points in time of a continuous time evolution.

\section{Preliminaries}
By $\mathcal{B}(\mathbb{C}^n)$ we denote the complex vector space of linear operators on $\mathbb{C}^n$. By $H(n)$  we denote the real vector space of hermitian operators on $\mathbb{C}^n$ and $H(n)_0$ denotes the subspace of $H(n)$ consisting of traceless hermitian operators. We equip both $H(n)$ and $\mathcal{B}(\mathbb{C}^n)$ with the Hilbert-Schmidt inner product. By $S_{H(n)_0}:=\{X\in H(n)_0:\ \|X\|_2=1\}$ we denote the unit sphere in $H(n)_0$ where $\|\cdot\|_2$ denotes the Hilbert-Schmidt norm. By $\mathcal{S}(\mathbb{C}^n)$ we denote the set of quantum states on $\mathbb{C}^n$, i.e. $\mathcal{S}(\mathbb{C}^n):=\{\rho\in H(n):\rho\geq 0, \text{tr}(\rho)=1\}$. Furthermore, for a subset $A\subseteq H(n)$, we denote by $\Delta(A)$ the set of differences of operators in $A$, i.e. $\Delta(A):=\{X-Y:X,Y\in A\}$. By $U(n)$ we denote the set of unitary operators on $\mathbb{C}^n$. We call a subset $A\subseteq \mathbb{R}^n$ an algebraic set if it is the real common zero locus of a set of real polynomials in $n$ variables and we call it a semi-algebraic set if it is the set of common solutions of a finite set of real polynomial inequalities in $n$ variables (cf. \cite{bochnak1998real}).

General quantum mechanical measurements can be described by positive operator valued measures (POVMs)\cite{holevo2011probabilistic, busch1995operational}. For the purpose of the present paper we use the following definition.
\begin{definition}(POVM.)
A POVM on $\mathbb{C}^n$ is a tuple $P=(Q_{1},\hdots,Q_{m})$ of positive semidefinite operators on $\mathbb{C}^n$ such that
\begin{align*}
\sum_{i=1}^{m}Q_{i}=\id_{\mathbb{C}^n}.
\end{align*}
An element of $P$ is called an effect operator. The number of outcomes of $P$ is $m$. 
\end{definition}
There is a linear map $h_P$ associated to each POVM $P=(Q_{1},\hdots,Q_{m})$ given by
\begin{align*}
h_{P}:H(n)&\to \mathbb{R}^{m} \\
  x&\mapsto\big( \text{tr}(Q_{1}x),\hdots,\text{tr}(Q_{m}x) \big).
\end{align*}
A whole experiment might consist of measuring more than one POVM. 
\begin{definition}(Measurement-scheme.)
A tuple of POVMs is called a measurement-scheme.
\end{definition}
For a tuple of natural numbers $I=(m_1,\hdots,m_l)$ let
\begin{align*}
\mathcal{M}(I):=\{(P^{i})_{i=1}^{l}:\ P^{i}\text{ is a POVM with }m_i\text{ outcomes.}\}. 
\end{align*}

Similar to a POVM, a measurement-scheme $M=(P^1,\hdots,P^k)$ induces a linear map
\begin{align*}
h_{M}:H(n)&\to \mathbb{R}^{|P^1|+\hdots+|P^k|} \\
  x&\mapsto\big( h_{P^1}(x),\hdots,h_{P^k}(x) \big).
\end{align*}
We equip $\mathcal{M}(I)$ with the topology induced by the metric
\begin{align*}
d(M,M^\prime):=\|h_M-h_{M^\prime}\|
\end{align*}
where $\|\cdot\|$ denotes the operator norm.

\begin{definition}\label{defcomplete}($\mathcal{R}$-complete.)
A measurement-scheme $M$ is called $\mathcal{R}$-complete for a subset $\mathcal{R}\subseteq \mathcal{S}(\mathbb{C}^n)$ iff $h_{M}|_{\mathcal{R}}$ is injective. An $\mathcal{S}(\mathbb{C}^n)$-complete POVM is called informationally complete.
\end{definition}
Furthermore, we use the following notion of stability of measurement-schemes (cf. Definition III.2 \cite{kech}).
\begin{definition}(Stability.)
Let $\mathcal{R}\subseteq \mathcal{S}(\mathbb{C}^n)$ be a subset and let $I$ be a tuple of natural numbers. An $\mathcal{R}$-complete measurement-scheme $M\in \mathcal{M}(I)$ is stably $\mathcal{R}$-complete iff there exists a neighbourhood $N\subseteq \mathcal{M}(I)$ of $M$ such that each measurement-scheme $M^\prime\in N$ is $\mathcal{R}$-complete.
\end{definition}
In case the subset $\mathcal{R}\subseteq \mathcal{S}(\mathbb{C}^n)$ is a smooth submanifold, the equivalence of this notion of stability to other stability properties is proven in \cite{kech}.

Finally, let us define the measurement-schemes we work with in the following. Let $l\in\mathbb{N}$. For $\mathcal{T}:\mathcal{B}(\mathbb{C}^n)\to \mathcal{B}(\mathbb{C}^n)$ a unital completely positive map and $P:=(Q_1,\hdots,Q_m)$ a POVM define the measurement-scheme 
\begin{align}\label{scheme}
\mathcal{T}^l(P):=\left(\left(Q_1,\hdots,Q_m\right),\left(\mathcal{T}(Q_1),\hdots,\mathcal{T}(Q_m)\right),\hdots,(\mathcal{T}^{l-1}(Q_1),\hdots,\mathcal{T}^{l-1}(Q_m))\right).
\end{align}
Here the POVM $P$ is understood to be the initial measurement set-up and $\mathcal{T}^l(P)$ is the measurement-scheme in which the POVM $P$ is measured at $l$ steps of the discrete time evolution described by completely positive trace preserving (CPTP) map  $\mathcal{T}^\dagger$.

\section{Unitary Time Evolution}\label{unitary}
\subsection*{Informationally Complete Tomography}
For given $U\in U(n)$ let
\begin{align*}
\mathcal{T}_U:\mathcal{B}(\mathbb{C}^n)&\to \mathcal{B}(\mathbb{C}^n)\\
x&\mapsto UxU^\dagger
\end{align*}
be the associated unital CPTP map. Furthermore, let $F_U$ be the fix point set of $\mathcal{T}_U$, i.e. $F_U:=\{X\in \mathcal{B}(\mathbb{C}^n):\mathcal{T}_U(X)=X\}$ and let let $F^{sa}_U:=\{X\in H(n):\mathcal{T}_U(X)=X\}$. Note that we have the block decomposition $\mathcal{T}_U=\mathcal{T}_U|_{F_U}\oplus \mathcal{T}_U|_{F_U^\bot}$.

We first deal with the problem of performing informationally complete quantum tomography using a given unitary time evolution.
\begin{proposition}\label{lower}
Let $P:=(P_1,\hdots,P_m)$ be a POVM and let $l\in\mathbb{N}$. If $\text{span}_{\mathbb{R}}\ \mathcal{T}_U^l(P)=H(n)$, then $m\geq n$.
\end{proposition}
\begin{proof}
Let $P:=(P_1,\hdots,P_m)$ be a POVM and assume $\text{span}_{\mathbb{R}}\mathcal{T}_U^l=H(n)$. Clearly $\dim F^{sa}_U\geq n$ because $F^{sa}_U$ contains the vector space of real matrices that are diagonal in a basis that diagonalizes $U$. 

Let $\Pi_{F^{sa}_U}:H(n)\to F^{sa}_U$ be the orthogonal projection on $F^{sa}_U$ and note that $\Pi_{F^{sa}_U}\circ \mathcal{T}_U^j(X)=\Pi_{F^{sa}_U}(X)$ for all $X\in H(n)$ and $j\in\mathbb{N}$. Therefore, $F^{sa}_U=\Pi_{F^{sa}_U}(H(n))=\Pi_{F^{sa}_U}(\text{span}_{\mathbb{R}}\mathcal{T}_U^l(P))=\text{span}_{\mathbb{R}}\{\Pi_{F^{sa}_U}(P_1),\hdots,\Pi_{F^{sa}_U}(P_m)\}$ and hence $\dim F^{sa}_U\leq m$.

Combining this, we conclude $m\geq \dim F^{sa}_U\geq n$.
\end{proof}
Thus, to allow for an informationally complete tomography, the POVM $P$ has to have at least $n$ outcomes. Furthermore, assuming that $P$ has $n+1$ outcomes, one has to measure at a minimum of $n$ points in time to achieve an informationally complete tomography and in the following we will see that this indeed suffices.

\begin{definition}(Feasible.)\label{def1}
A unitary matrix  $U\in U(n)$ is feasible iff the algebraic multiplicity of each eigenvalue of $\mathcal{T}_U|_{F_U^\bot}$ is one and $\dim F_U=n$.
\end{definition}

Let us note that almost all unitaries are feasible.

\begin{theorem}\label{thmUfull}(Informationally complete tomography.)
Let $U\in U(n)$ be feasible. Then, for almost all POVMs $P$ with $n+1$ outcomes the measurement scheme $\mathcal{T}_U^n(P)$ is informationally complete.
\end{theorem}
\begin{remark}
The previous statement also holds when measuring POVMs with $n$ outcomes at $n+1$ points in time.
\end{remark}
The proof of this result can be found in Subsection \ref{proof1}.

\subsection*{Tomography under Prior Information}

In this subsection we extend the results of the previous subsection to quantum tomography on subsets $\mathcal{R}\subseteq\mathcal{S}(\mathbb{C}^n)$. 
\begin{definition}
Let $\mathcal{R}\subseteq \mathcal{S}(\mathbb{C}^n)$ be a subset. A semi-algebraic set $\mathcal{D}\subseteq H(n)$ with $0\notin\mathcal{D}$ represents $\Delta(\mathcal{R})$ iff for every measurement-scheme $M$ with  $h_M(X)=0$ for some $X\in\Delta(\mathcal{R})-\{0\}$ there exists $Y\in\mathcal{D}$ such that $h_M(Y)=0$.
\end{definition}
\begin{remark}
Note that a measurement-scheme $M$ is not $\mathcal{R}$-complete if and only if there exists $X\in\Delta(\mathcal{R})-\{0\}$ such that $h_M(X)=0$. Thus, the set of measurement-schemes that solve the equation $h_M(Y)=0$ for some $Y\in\mathcal{D}$ contains the set of measurement-schemes that are not $\mathcal{R}$-complete.
\end{remark}

The next theorem is the main result of this section. It asserts that the measurement-scheme $\mathcal{T}^l_U(P)$ defined in Equation \eqref{scheme} is suited to perform tomography on arbitrary semi-algebraic subsets $\mathcal{R}\subseteq\mathcal{S}(\mathbb{C}^n)$.
\begin{theorem}\label{thmUprior}(Universality.)
Let $U\in U(n)$ be feasible. For $\mathcal{R}\subseteq \mathcal{S}(\mathbb{C}^n)$ a subset, let $\mathcal{D}$ be a semi-algebraic set that represents $\Delta(\mathcal{R})$. Let $m\ge n$ and let $l\in\mathbb{N}$ be such that $l(m-1)>\dim \mathcal{D}$. Then, for almost all POVMs $P$ with $m$ outcomes the measurement-scheme $\mathcal{T}_U^l(P)$ is stably $\mathcal{R}$-complete.
\end{theorem}
Note that this statement cannot hold for all POVMs as the POVM $P:=1/n(\mathds{1}_n,\hdots,\mathds{1}_n)$ is a counterexample. The proof of this theorem can be found in Subsection \ref{proof2}. 

In the following we discuss some consequences of Theorem \ref{thmUprior}. First, it directly implies a Whitney type embedding result.
\begin{corollary}(Whitney.)
Let $U\in U(n)$ be feasible and let $\mathcal{R}\subseteq \mathcal{S}(\mathbb{C}^n)$ be a subset. Let $m\ge n$ and let $l\in\mathbb{N}$ be such that $l(m-1)>2\dim \mathcal{R}$. Then, for almost all POVMs $P$ with $m$ outcomes the measurement scheme $\mathcal{T}_U^l(P)$ is stably $\mathcal{R}$-complete.
\end{corollary}
\begin{proof}
Assume w.l.o.g. that $\mathcal{R}$ is algebraically closed, because if not we can replace $\mathcal{R}$ by its algebraic closure without changing the dimension (see Proposition 2.8.2 of \cite{bochnak1998real}). By the proof of Lemma IV.2 of \cite{kech2}, $\Delta(\mathcal{R})-\{0\}$ is a semi-algebraic set with $\dim \Delta(\mathcal{R})-\{0\}\le 2\dim\mathcal{R}$. Applying Theorem \ref{thmUprior} to $\Delta(\mathcal{R})-\{0\}$ concludes the proof.
\end{proof}

Theorem \ref{thmUprior} can also be applied to tomography on states of bounded rank. Let $\mathcal{S}_r(\mathbb{C}^n):=\{\rho\in\mathcal{S}(\mathbb{C}^n):\text{rank}(\rho)\leq r\}$ be the set of quantum states of rank at most $r$.

\begin{corollary}(Tomography on states of bounded rank.)\label{corBounded}
Let $U\in U(n)$ be feasible. Let $m\ge n$ and let $l\in\mathbb{N}$ be such that $l(m-1)\geq 4r(n-r)-1$. Then, for almost all POVMs $P$ with $m$ outcomes the measurement-scheme $\mathcal{T}_U^l(P)$ is stably $\mathcal{S}_r(\mathbb{C}^n)$-complete.
\end{corollary}
\begin{proof}
The proof follows directly from applying  Theorem \ref{thmUprior} to the algebraic set $\mathcal{D}$ defined in Lemma IV.6 of \cite{kech2}.
\end{proof}
\begin{remark}
If a measurement-scheme is $\mathcal{S}_r^n$-complete, it was shown in \cite{heinosaari2013quantum,kech} that, up to terms at most logarithmic in $n$, we have $m\ge 4r(n-r)$ and in this sense the lower bound given in Corollary \ref{corBounded} is nearly optimal. 
\end{remark}
Let us note that similar to Corollary V.12 of \cite{kech2}, Corollary \ref{corBounded} implies corresponding results for tomography on states of fixed spectrum. 

\section{CPTP Time Evolution}
In this section we generalize the scenario of Section \ref{unitary} by considering a larger class of system dynamics. With this generalization the lower bound on the dimension of the initial POVM as given by Proposition \ref{lower} can be relaxed. Indeed we show that one dimensional POVMs can suffice for informationally complete tomography and that they are also suited for tomography on subsets $\mathcal{R}\subseteq\mathcal{S}(\mathbb{C}^n)$.

\begin{definition}(Feasible.)\label{def2}
A CPTP map $\mathcal{T}$ is feasible iff it is invertible and the algebraic multiplicity of each of its eigenvalues is one.
\end{definition}

The following result is the main result of this section. It is a universality result analogous to Theorem \ref{thmUprior}.
\begin{theorem}(Universality.)\label{thmUCPTP}
Let $\mathcal{T}$ be a feasible CPTP map. For $\mathcal{R}\subseteq\mathcal{S}(\mathbb{C}^n)$ a subset, let $\mathcal{D}$ be a semi-algebraic set that represents $\Delta(\mathcal{R})$. Furthermore, let $m\in\mathbb{N}$ and let $l\in\mathbb{N}$ be such that $l(m-1)>\dim\mathcal{D}$. Then, for almost all POVMs $P$ with $m$ outcomes the measurement-scheme $(\mathcal{T}^\dagger)^l(P)$ is stably $\mathcal{R}$-complete.
\end{theorem}
The proof of this theorem can be found in Subsection \ref{proof3}. An immediate consequence is the case $k=1$ which may be of particular interest as it shows that in fact an initial POVM of minimal dimension suffices to perform  tomography on arbitrary subsets of the state space.
\begin{corollary}\label{cor1}
Let $\mathcal{T}$ be a feasible CPTP map. For $\mathcal{R}\subseteq\mathcal{S}(\mathbb{C}^n)$ a subset, let $\mathcal{D}$ be a semi-algebraic set that represents $\Delta(\mathcal{R})$. Furthermore, let $l\in\mathbb{N}$ be such that $l>\dim\mathcal{D}$. Then, for almost all POVMs $P$ with two outcomes the measurement-scheme $(\mathcal{T}^\dagger)^{l}(P)$ is stably $\mathcal{R}$-complete.
\end{corollary}
\begin{remark}
Given a POVM with two outcomes $P:=\{P_1,P_2\}$, all the relevant information is contained in $P_1$ as $P_2=\id_{\mathbb{C}^n}-P_1$. In this sense one can identify a one dimensional POVM with an observable $O:=P_1$. Under this identification the measurement-scheme $\mathcal{T}^{l}(P)$ corresponds to measuring the expectation value of $O$ at $l$ time steps of the time evolution given by $\mathcal{T}$. Corollary \ref{cor1} then states that any two states of a given subset $\mathcal{R}\subseteq\mathcal{S}(\mathbb{C}^n)$ can be discriminated by determining the expectation value of a single observable at sufficiently many time steps.
\end{remark}
In the remainder of this section we give some further corollaries of Theorem \ref{thmUCPTP}. Of course Theorem \ref{thmUCPTP} also covers the case of informationally complete tomography.
\begin{corollary}(Informationally Complete Tomography)
Let $\mathcal{T}$ be a feasible CPTP map. Furthermore, let $m\in\mathbb{N}$ and let $l\in\mathbb{N}$ be such that $l(m-1)\geq n^2-1$. Then, for almost all POVMs $P$ with $m$ outcomes the measurement-scheme $(\mathcal{T}^\dagger)^{l}(P)$ is informationally complete.
\end{corollary}
\begin{proof}
Note that $S_{H(n)_0}$ represents $\Delta(\mathcal{S}(\mathbb{C}^n))$. Applying Theorem \ref{thmUCPTP} to $S_{H(n)_0}$, together with the observation that $\dim S_{H(n)_0}=n-2$, concludes the proof.
\end{proof}

Another immediate consequence is a Whitney type embedding result.
\begin{corollary}(Whitney.)
Let $\mathcal{T}$ be a feasible CPTP map and let $\mathcal{R}\subseteq\mathcal{S}(\mathbb{C}^n)$ be a subset. Furthermore, let $m\in\mathbb{N}$ and let $l\in\mathbb{N}$ be such that $l(m-1)>2\dim\mathcal{R}$. Then for almost all POVMs $P$ with $m$ outcomes the measurement-scheme $(\mathcal{T}^\dagger)^l(P)$ is stably $\mathcal{R}$-complete.
\end{corollary}
\begin{proof}
Assume w.l.o.g. that $\mathcal{R}$ is algebraically closed, because if not we can replace $\mathcal{R}$ by its algebraic closure without changing the dimension (see Proposition 2.8.2 of \cite{bochnak1998real}). By the proof of Lemma IV.2 of \cite{kech2}, $\Delta(\mathcal{R})-\{0\}$ is a semi-algebraic set with $\dim \Delta(\mathcal{R})-\{0\}\le 2\dim\mathcal{R}$. Applying Theorem \ref{thmUCPTP} to $\Delta(\mathcal{R})-\{0\}$ concludes the proof.
\end{proof}
Finally, Theorem \ref{thmUCPTP} can also be straightforwardly applied to tomography on states of bounded rank.
\begin{corollary}(Tomography on states of bounded rank.)
Let $\mathcal{T}$ be a feasible CPTP map. Furthermore, let $m\in\mathbb{N}$ and let $l\in\mathbb{N}$ be such that  $l(m-1)\geq 4r(n-r)-1$. Then, for almost all POVMs $P$ with $m$ outcomes the measurement scheme $(\mathcal{T}^\dagger)^l(P)$ is stably $\mathcal{S}_r(\mathbb{C}^n)$-complete.
\end{corollary}
\begin{proof}
The proof follows directly from applying Theorem \ref{thmUCPTP} to the algebraic set $\mathcal{D}$ defined in Lemma IV.6 of \cite{kech2}.
\end{proof}

\section{Proofs of Technical Results}
The proofs of the following results are all based on the approach presented in \cite{kech2}: Let $\mathcal{R}\subseteq\mathcal{S}(\mathbb{C}^n)$ be a subset. Among all admissible measurement-schemes we characterize the subset $N$ of non $\mathcal{R}$-complete measurement-schemes by real algebraic equations. We then prove that the subset $N$ has a smaller dimension than the set of all admissible measurement-schemes, showing that almost all admissible measurement-schemes are $\mathcal{R}$-complete.

Denote by $\mathcal{P}(m)$ the set of POVMs with $m+1$ outcomes. Let us begin by briefly discussing the measure we choose on $\mathcal{P}(m)$. Via the injective mapping 
\begin{align*}
\eta:\mathcal{P}(m)\to (H(n))^m,\ (Q_1,\hdots,Q_{m+1})\to (Q_1,\hdots,Q_{m})
\end{align*}
we can identify $\mathcal{P}(m)$ with the subset $\eta(\mathcal{P}(m))$ of $(H(n))^m$. The measure we choose on $\mathcal{P}(m)$ is the Lebesgue measure it inherits when identified with the subset $\eta(\mathcal{P}(m))\subseteq(H(n))^m$. 
\subsection{Proof of Theorem \ref{thmUfull}}\label{proof1}
The proof of this theorem serves as a blueprint for the other proofs presented in this section. Therefore, let us begin by giving a short outline of the proof to make our argument more transparent. Let $K=(H(n))^n$. Furthermore, observe that $\Delta(S(\mathbb{C}^n))\subseteq H(n)_0$ and thus $H(n)_0-\{0\}$ represents $\Delta(S(\mathbb{C}^n))$.

For $i\in\{1,\hdots,n\},\ j\in\{0,\hdots,n-1\}$ define real polynomials 

\begin{align*}
p_{i,j}:&K\times H(n)_0\simeq \mathbb{R}^{n^3}\times \mathbb{R}^{n^2-1} \to\mathbb{R},\\
&(P,X)\mapsto \text{tr}\left((\mathcal{T}_U)^j(Q_i)X\right)=\text{tr}\left(Q_i(\mathcal{T}_U^\dagger)^j(X)\right).
\end{align*}
Denote by $\mathcal{V}$ the real common zero locus of the set of polynomials $\{p_{i,j}\}_{i\in\{1,\hdots,n\},\ j\in\{0,\hdots,n-1\}}$ and let 
\begin{align*}
\mathcal{M}:=\left(K\times (H(n)_0-\{0\})\right)\cap\mathcal{V}.
\end{align*}
Clearly, $\mathcal{M}$ is a semi-algebraic set and furthermore it characterizes the $n$-dimensional POVMs $P$ for which $\mathcal{T}_U^n(P)$ is not informationally complete in the following sense: Let $\pi_1:K\times H(n)_0\to K$ denote the projection on the first factor $K$. Let 
\begin{align*}
K_{NC}:=\{P\in \mathcal{P}(n): \mathcal{T}_U^n(P)\ \text{is not informationally complete}\}.
\end{align*}
Then, since $\eta(\mathcal{P}(n))$ is a subset of $K$ and $H(n)_0-\{0\}$ represents $\Delta(S(\mathbb{C}^n))-\{0\}$, we have $\eta(K_{NC})\subseteq\pi_1(\mathcal{M})$. We show in the following that $\dim\mathcal{M}<\dim K=n^3$. But then, by Theorem 2.8.8 of \cite{bochnak1998real}, we have $\dim\pi_1(\mathcal{M})<\dim K=n^3$ and thus $\pi_1(\mathcal{M})$ has measure zero in $K$. Since $\eta(K_{NC})$ is a subset of $\pi_1(\mathcal{M})$, we finally conclude that $\eta(K_{NC})$ also has measure zero in $K$.

As a first step, we construct a decomposition of $H(n)_0$ with respect to the eigenstates of $\mathcal{T}_U^\dagger$ which allows us to simplify the analysis in the following. By changing the basis if necessary, we can assume $U$ to be a diagonal matrix. Let $\{\lambda_{ij}\}_{i,j\in\{1,\hdots,n\}}$ be the multiset of eigenvalues of $\mathcal{T}^\dagger_U$ such that for all $i,j\in\{1,\hdots,n\}$ we have 
\begin{align*}
\mathcal{T}_U^\dagger(e_{ij})=\lambda_{ij}e_{ij},
\end{align*}
where $e_{ij}$ denotes the matrix whose only non-vanishing entry is a $1$ in the $i$-th row and $j$-th column. Note that if $U=\text{diag}(\lambda_1,\hdots,\lambda_n)$, then $\lambda_{ij}=\lambda_i^*\lambda_j$ for all $i,j\in\{1,\hdots,n\}$. For $k\in\mathbb{N}_0$ let 
\begin{align*}
E_{2k}:=\{X\in H(n)_0: \text{tr}(Xe_{ij})\neq 0\text{ for at most }2k\text{ pairs }(i,j),\ i\neq j\}.
\end{align*}
Note that $E_0=F_U\cap H(n)_0$ since $U$ is feasible by assumption.
\begin{proposition}\label{propE}
$E_{2k}$ is an algebraic set with $\dim E_k= n-1+2k$.
\end{proposition}
\begin{proof}
Let $C:=\{(i,j)\in \mathbb{N}^2:i> j\}$. For $S\subseteq C$ let $N(S):=\{X\in H(n)_0:\ \text{tr}(Xe_s)=0,\ \forall s\in S\}$. $N(S)$ is a linear subspace and thus clearly an algebraic set. Furthermore, let $(j,l)\in C$ and note that $e_{jl}=(e_{jl}+e_{lj})+i(-ie_{jl}+ie_{lj})$. Thus, for a hermitian matrix $X\in H(n)_0$ we have $\text{tr}(Xe_{jl})=0$ if and only if $\text{tr}\left(X(e_{jl}+e_{lj})\right)$ and $\text{tr}\left(X(ie_{jl}-ie_{lj})\right)=0$. The matrices $\{e_{jl}+e_{lj},ie_{jl}-ie_{lj}\}_{(j,l)\in C}\subseteq H(n)_0$ are linearly independent and hence $\dim N(S)=n^2-1-2|S|$. But $E_{2k}=\bigcup_{S\subseteq C, |S|=\frac{n^2-n}{2}-k}N(S)$. Hence, as a finite union of algebraic sets, $E_{2k}$ is an algebraic set and furthermore $\dim E_k=n^2-1-2(\frac{n^2-n}{2}-k)= n-1+2k$.
\end{proof}
Let $R_0:=E_0-\{0\}$. For $k\in\{1,\hdots,\lceil n/2\rceil-1\}$ let $R_{k}:=E_{2k}-E_{2k-2}$ be the set of hermitian matrices with precisely $2k$ non-vanishing off-diagonal entries and let $R_{\lceil n/2\rceil}:=H(n)_0-E_{2\lceil n/2\rceil-2}$. Observe that $R_{\lceil n/2\rceil}$ might be empty and that $H(n)_0-\{0\}=\bigcup_{k=0}^{\lceil n/2\rceil}R_k.$

Furthermore, for $k\in\{0,\hdots,\lceil n/2\rceil\}$, let $\mathcal{M}_k:=(K\times R_k)\cap \mathcal{V}$ and observe that 
\begin{align*}
\mathcal{M}=\bigcup_{k=0}^{\lceil n/2\rceil}\mathcal{M}_k.
\end{align*}
Hence, to complete the proof of Theorem \ref{thmUfull}, it suffices to prove the following proposition.
\begin{proposition}
Let $k\in\{0,\hdots,\lceil n/2\rceil\}$. Then $\dim \mathcal{M}_k<\dim K$.
\end{proposition}
\begin{proof}
Let $R^1_k:=\{X\in R_k:\ \text{tr}(XE_0)=0\}$ and $R^2_k=R_k-R^1_k$. Note that $R^1_0=\emptyset$. Going along the lines of the proof of Proposition \ref{propE}, it is seen that that $\dim R^1_k=2k$. For $i\in\{1,2\}$ let $\mathcal{M}^i_k:=(K\times R_k^i)\cap \mathcal{V}$
and note that $\mathcal{M}_k= \mathcal{M}^1_k\cup \mathcal{M}^2_k$. 

In a first step, we prove that $\dim\mathcal{M}^2_k<\dim K$. In order to do so we prove the following proposition as an intermediate step.
\begin{proposition}\label{propind}
For $X\in R^2_k$, the set of matrices $\{(\mathcal{T}_U^\dagger)^j(X)\}_{j\in\{0,\hdots,2k\}}$ is linearly independent over $\mathbb{C}$.
\end{proposition}
\begin{proof}
First note that by construction of $R_k^2$ there is an $i\in\{1,\hdots,n\}$ such that $0\neq\text{tr}(Xe_{ii})=:X_{0}$. Furthermore, by construction of $R_k$, there are distinct eigenvectors $e_{i_1j_1},\hdots,e_{i_{2k}j_{2k}}$ such that for all $m\in\{1,\hdots,2k\}$ we have $i_m\neq j_m$ and $0\neq\text{tr}(Xe_{i_mj_m})=:X_m$. Let  $Y=\text{span}_{\mathbb{C}}(\{e_{ii}\}\cup\{e_{i_mj_m}\}_{m\in\{1,\hdots,2k\}})$  and let $\pi_Y:\mathcal{B}(\mathbb{C}^n)\to Y$ be the orthogonal projection on $Y$. It is enough to show that $\left\{\pi_{Y}\left((\mathcal{T}_U^\dagger)^j(X)\right)\right\}_{j\in\{0,\hdots,2k\}}$ 
is a set linearly independent operators over $\mathbb{C}$. Consider the $(2k+1)\times (2k+1)$ matrix $M:=\left(\text{tr}\left(\pi_{Y}\left((\mathcal{T}_U^\dagger)^l(X)\right)e_{i_{m}j_{m}}\right)\right)_{l,m=0}^{2k}$, where $e_{i_0j_0}=e_{ii}$. The determinant of $M$ is proportional to the determinant of the Vandermonde matrix whose entries are determined by the eigenvalues of $\mathcal{T}_U^\dagger$:
\begin{align*}
&\det(M)\\
&=\det\begin{bmatrix}
X_0 & X_0 & \dots & X_0 \\
\lambda_{i_1j_1}^0X_1 & \lambda_{i_1j_1}^1X_1 & \dots & \lambda_{i_1j_1}^{2k}X_1 \\
\vdots & \vdots &  & \vdots \\
\lambda_{i_{2k}j_{2k}}^0X_{2k} & \lambda_{i_{2k}j_{2k}}^1X_{2k} & \dots& \lambda_{i_{2k}j_{2k}}^{2k}X_{2k}
\end{bmatrix}
=X_0\cdot \hdots\cdot X_{2k}\ \det\begin{bmatrix}
1 & 1 & \dots & 1 \\
\lambda_{i_1j_1}^0 & \lambda_{i_1j_1}^1 & \dots & \lambda_{i_1j_1}^{2k} \\
\vdots & \vdots &  & \vdots \\
\lambda_{i_{2k}j_{2k}}^0 & \lambda_{i_{2k}j_{2k}}^1 & \dots& \lambda_{i_{2k}j_{2k}}^{2k}
\end{bmatrix}\\
\\
&=\frac{X_0\cdot \hdots\cdot X_{2k}}{\lambda_{i_1j_1}\cdot\hdots\cdot\lambda_{i_{2k}j_{2k}}}\ \det\begin{bmatrix}
1 & 1 & \dots & 1 \\
\lambda_{i_1j_1}^1 & \lambda_{i_1j_1}^2 & \dots & \lambda_{i_1j_1}^{2k+1} \\
\vdots & \vdots &  & \vdots \\
\lambda_{i_{2k}j_{2k}}^1 & \lambda_{i_{2k}j_{2k}}^2 & \dots& \lambda_{i_{2k}j_{2k}}^{2k+1}
\end{bmatrix}\\
&=\frac{X_0\cdot \hdots\cdot X_{2k}}{\lambda_{i_1j_1}\cdot\hdots\cdot\lambda_{i_{2k}j_{2k}}}\prod_{0\leq o< p\leq m}(\lambda_{i_pj_p}-\lambda_{i_oj_0}),
\end{align*}
where $\lambda_{i_0j_0}=1$. Since $U$ is a unitary matrix, we have $\lambda_{i_mj_m}\neq0$ for all $m\in\{1,\hdots,2k\}$. Furthermore, since $U$ is feasible by assumption, $\lambda_{i_lj_l}\neq\lambda_{i_mj_m}$ for all $m,l\in\{0,\hdots,2k\}$ with $m\neq l$. This, together with $X_m\neq 0,\ m\in\{0,\hdots,2k\}$, shows that $\det(M)\neq 0$ and hence proves the claim.
\end{proof}

Now let $X\in R_k^2$ be fixed. From the previous proposition we conclude that at least $n\cdot\min\{(2k+1),n\}$ of the linear equations
\begin{align}\label{eq2}
\text{tr}\left(Q_i(\mathcal{T}^\dagger_U)^j(X)\right)=0,\ i\in\{1,\hdots,n\},\ j\in\{0,\hdots,n-1\},\ (Q_1,\hdots,Q_n)\in K,
\end{align}
are independent. Hence the dimension of the solution set of the equations \eqref{eq2} is by at least $n\cdot\min\{2k+1,n\}$ smaller than the dimension of $K$. Since this holds for all $X\in R_k^2$ we find 
\begin{align*}
\dim \mathcal{M}^2_k&\leq \dim K+\dim R_k^2-n(2k+1)= \dim K+n-1+2k-n(2k+1)\\
&\leq \dim K-2(n-1)k-1<\dim K
\end{align*}
if $k\in\{0,\hdots,[n/2]-1\}$ using Proposition \ref{propE} and furthermore 
\begin{align*}
\dim \mathcal{M}^2_{[n/2]}&\leq \dim K+\dim H(n)_0-n^2=\dim K+n^2-1-n^2\\
&= \dim K-1<\dim K.
\end{align*}

In the second step we show that $\dim\mathcal{M}_k^1<n^3$. Note that we can restrict to $k>0$ since $R^1_0=\emptyset$. Let $X\in R^1_k$ be fixed. By going along the lines of the proof of Proposition \ref{propind} it follows that the smaller set of operators $\{(\mathcal{T}_U^\dagger)^j(X)\}_{j\in\{0,\hdots,2k-1\}}$ still is linearly independent over $\mathbb{C}$. Considering $X\in R^1_k$ in the proof of Proposition \ref{propind} just corresponds to setting $X_0=0$. The remainder of the argument still applies. We conclude that at least $n\cdot\min\{2k,n\}$ of the linear equations
\begin{align}\label{eqp}
\text{tr}\left(Q_i(\mathcal{T}^\dagger_U)^j(X)\right)=0,\ i\in\{1,\hdots,n\},\ j\in\{0,\hdots,n-1\},\ (Q_1,\hdots,Q_n)\in K,
\end{align}
are independent. Thus, the dimension of the solution set of the equations \eqref{eqp} is by at least $n\cdot\min\{2k,n\}$ smaller than the dimension of $K$. Since this holds for all  $X\in R^1_k$ we find
\begin{align*}
\dim \mathcal{M}^1_k&\leq \dim K+\dim R_k^1-n(2k)= \dim K+2k-2nk\\
&\leq \dim K-2(n-1)k<\dim K
\end{align*}
if $k\in\{1,\hdots,\lceil n/2\rceil-1\}$ and furthermore 
 \begin{align*}
 \dim \mathcal{M}^1_{\lceil n/2\rceil}&\leq \dim K+\dim H(n)_0-n^2=\dim K+n^2-1-n^2\\
 &= \dim K-1<\dim K. 
 \end{align*}

Finally, since $\mathcal{M}_k= \mathcal{M}^1_k\cup \mathcal{M}^2_k$, we conclude $\dim \mathcal{M}_k<\dim K$.
\end{proof}

\subsection{Proof of Theorem \ref{thmUprior}}\label{proof2}
Just like in the last subsection, we identify the set of $m$-dimensional POVMs on $\mathbb{C}^n$ with the semi-algebraic subset $\eta(\mathcal{P}(m))$ of $(H(n))^m$. Let $K=(H(n))^m$. Define the semi-algebraic map $\phi$ by
\begin{align}\label{phi}
\begin{split}
\phi:H(n)_0-\{0\}&\to S_{H(n)_0}\\
x&\mapsto \frac{x}{\|x\|_2}.
\end{split}
\end{align}
Let $\mathcal{D}$ be the semi-algebraic set that represents $\Delta(\mathcal{R})$, then $\phi(\mathcal{D})$ represents $\Delta(\mathcal{R})$ by Proposition 2.2.7  of \cite{bochnak1998real} and $\dim\phi(\mathcal{D})\leq\dim\mathcal{D}$ by Theorem 2.8.8 of \cite{bochnak1998real}. $S_{H(n)_0}$ is closed in the norm topology and by Proposition 2.8.2 of \cite{bochnak1998real} the dimension of a semi-algebraic set coincides with the dimension of its closure in the norm topology. We conclude that the closure $\overline{\phi(\mathcal{D})}$ of $\phi(\mathcal{D})$ is a subset of $S_{H(n)_0}$ with $\dim\overline{\phi(\mathcal{D})}\leq\dim\mathcal{D}$. In addition, by Proposition 2.2.2 of \cite{bochnak1998real}, $\overline{\phi(\mathcal{D})}$ is semi-algebraic and hence represents $\Delta(\mathcal{R})$. In the following we replace $\mathcal{D}$ by $\overline{\phi(\mathcal{D})}$ if $\mathcal{D}$ is not a closed subset of $S_{H(n)_0}$. 

For the most part, the remainder of this proof can be straightforwardly obtained by going along the lines of the proof of Theorem \ref{thmUfull}.  However, for the sake of completeness, let us give the whole argument.

For $k\in\{0,1,\hdots,\lceil l/2\rceil-1\}$, let $\mathcal{D}^i_k:=\mathcal{D}\cap R^i_k,\ i=1,2,$ and note that we have $\dim \mathcal{D}^1_k\leq 2k-1$(if $k>1$) and $\dim \mathcal{D}^2_k\leq n-2+2k$. We get the upper bounds $2k-1$ and $n-2+2k$ rather then $2k$ and $n-1+2k$. This is because $\mathcal{D}\subseteq S_{H(n)_0}$ and $\dim (E_k\cap S_{H(n)_0})=n-2+2k+1$ as can be seen from the proof of propositions \ref{propE} and \ref{D}.  Let  $\mathcal{D}_{[l/2]}:=\mathcal{D}-\bigcup_{k=0}^{[l/2]-1}(\mathcal{D}^{1}_{k}\cup\mathcal{D}^{2}_{k})$ and let $\mathcal{D}^{1}_{[l/2]}:=\{X\in\mathcal{D}_{[l/2]}:\text{tr}(XE_0)=0\}$, $\mathcal{D}_{[l/2]}^2=\mathcal{D}_{[l/2]}-\mathcal{D}_{[l/2]}^1$.  Note that $\dim \mathcal{D}_{[l/2]}^i\le\dim\mathcal{D}$ for $i=1,2$. Also note that $\mathcal{D}=\bigcup_{j=1}^{[l/2]}(\mathcal{D}_{j}^1\cup\mathcal{D}_{j}^2)$.

Just like in the proof of Theorem \ref{thmUfull}, for $j\in\{0,\hdots,l-1\},\ i\in\{1,\hdots,m\}$, define real polynomials 
\begin{align*}
p_{i,j}:&K\times H_0(n)\simeq\mathbb{R}^{mn^2}\times\mathbb{R}^{n^2-1}\to\mathbb{R},\\
&(P,X)\mapsto\text{tr}\left(Q_i(\mathcal{T}_U^\dagger)^j(X)\right).
\end{align*}
Denote by $\mathcal{V}$ the common zero locus of the polynomials $\{p_{i,j}\}_{j\in\{0,\hdots,l-1\},\ i\in\{1,\hdots,m\}}$
and for $i\in\{1,2\}$, $k\in\{0,1,\hdots,\lceil l/2\rceil\}$,  let $\mathcal{M}^{i}_k:=(K\times \mathcal{D}^i_k)\cap\mathcal{V}$.

First we prove that for all $k\in\{0,1,\hdots,\lceil l/2\rceil\}$ we have $\dim\mathcal{M}^2_k<\dim K$. So let $k\in\{0,\hdots,\lceil l/2\rceil\}$. If $\mathcal{D}_k^2=\emptyset$ we have $\mathcal{M}^2_k=\emptyset$ and thus we clearly have $\dim\mathcal{M}^2_k<\dim K$. Otherwise let $X\in \mathcal{D}_k^2$ be fixed. From Proposition \ref{propind} we conclude that at least $m\cdot \min\{2k+1,l\}$ of the linear equations
\begin{align}\label{eqk}
\text{tr}\left(Q_i(\mathcal{T}^\dagger_U)^j(X)\right)=0,\ j\in\{0,\hdots,l-1\},\ i\in\{1,\hdots,m\},\ (Q_1,\hdots,Q_m)\in K,
\end{align}
are independent. Hence the dimension of the solution set of the equations \eqref{eqk} is by at least $m\cdot\min\{2k+1,l\}$ smaller than the dimension of $K$. Since this holds for all $X\in \mathcal{D}_k^2$ we find
\begin{align*}
\dim \mathcal{M}^{2}_k&\leq \dim K+\dim\mathcal{D}_k^2-m(2k+1)\le \dim K+n-2+2k-m(2k+1)
\\&=\dim K-(m-(n-1))-2k(m-1)-1<\dim K
\end{align*}
for $k\in\{0,\hdots,[l/2]-1\}$, using the assumption that $m\ge n-1$. Also by assumption we have $\dim\mathcal{D}<ml$ and thus  $\dim \mathcal{D}^2_{[l/2]}\leq\dim\mathcal{D}<ml$. Hence we conclude that
\begin{align*}
\dim \mathcal{M}_{[l/2]}^2&\le \dim K+\dim\mathcal{D}-ml\\
&<\dim K+ml-ml=\dim K.
\end{align*} 

Next we prove that for all $k\in\{1,\hdots,\lceil l/2\rceil\}$ we have $\dim\mathcal{M}^1_k<\dim K$. So let $k\in\{1,\hdots,\lceil l/2\rceil\}$. If $\mathcal{D}_k^1=\emptyset$ we have $\mathcal{M}^1_k=\emptyset$ and thus clearly $\dim\mathcal{M}^1_k<\dim K$. Otherwise let $X\in \mathcal{D}_k^1$ be fixed. From Proposition \ref{propind} we conclude that at least $m\cdot\min\{2k,l\}$ of the linear equations
\begin{align}\label{eql}
\text{tr}\left(Q_i(\mathcal{T}_U)^j(X)\right)=0,\ j\in\{0,\hdots,l-1\},\ i\in\{1,\hdots,m\},\ (Q_1,\hdots,Q_m)\in K,
\end{align}
are independent. Hence the dimension of the solution set of the equations \eqref{eql} is by at least $m\cdot\min\{2k,l\}$ smaller than the dimension of $K$. Since this holds for all $X\in\mathcal{D}^{1}_k$  we find
\begin{align*}
\dim \mathcal{M}^{1}_k&\leq \dim K+\dim\mathcal{D}_k^1-m(2k)=\dim K+2k-1-2mk\\
&=\dim K-2k(m-1)-1<\dim K
\end{align*}
for $k\in\{1,\hdots,[l/2]-1\}$. Using $\dim \mathcal{D}^1_{[l/2]}\leq\dim\mathcal{D}<ml$ we also conclude that 
\begin{align*}
\dim \mathcal{M}^1_{[l/2]}&\le \dim K+\dim\mathcal{D}-ml\\
&<\dim K+ml-ml=\dim K. 
\end{align*}
Finally, let $\pi_1:K\times H(n)_0\to K$ be the projection on the first factor $K$ and let $\mathcal{M}:=\bigcup_{k=0}^{\lceil l/2\rceil}(\mathcal{M}^1_k\cup\mathcal{M}^2_k)$. Clearly, $\mathcal{M}$ is a semi-algebraic set. Let 
\begin{align*}
K_{\mathcal{R}}:=\{P\in \mathcal{P}(m): \mathcal{T}_U^l(P)\ \text{ is not }\mathcal{R}-\text{complete.}\}.
\end{align*}
Then, since $\eta(\mathcal{P}(m))$ is a subset of $K$ and $\mathcal{D}$ represents $\Delta(\mathcal{R})$, we have $\eta(K_{\mathcal{R}})\subseteq\pi_1(\mathcal{M})$. We have shown that for all $i\in\{1,2\}$, $k\in\{0,\hdots,\lceil l/2\rceil\}$ we have $\dim\mathcal{M}^i_k<\dim K$ and thus $\dim \mathcal{M}<\dim K$. Hence we find $\dim\pi_1(\mathcal{M})<\dim K$ by Theorem 2.8.8 of \cite{bochnak1998real}. But since $\eta(K_{\mathcal{R}})$ is a subset of $\pi_1(\mathcal{M})$, we conclude that $\eta(K_{\mathcal{R}})$ has measure zero in $K$.

Finally, since $\mathcal{D}$ is a closed subset of $S_{H(n)_0}$, the stability follows form Lemma IV.1 of \cite{kech2}.

\subsection{Proof of Theorem \ref{thmUCPTP}}\label{proof3}
Just like in the proof of Theorem \ref{thmUprior}, we replace $\mathcal{D}$ by $\overline{\phi(\mathcal{D})}$ if $\mathcal{D}$ is not a closed subset of $S_{H(n)_0}$.

Again, the remainder of this proof is close to the proof of Theorem \ref{thmUfull}. Let $\{\id,e_1,\hdots,e_{n^2-1}\}$ be the set of eigenvectors of $\mathcal{T}^{\dagger}$. For $k\in\mathbb{N}_0$ let 
\begin{align*}
D_k:=\{X\in S_{H(n)_0}:\text{tr}(Xe_i)\neq 0\text{ for at most }k\text{ elements } i\in\{1,\hdots,n^2-1\}\}.
\end{align*}
\begin{proposition}\label{D}
For $k\in\{1,\hdots,n^2-1\}$, $D_k$ is an algebraic set of dimension at most $k-1$.
\end{proposition}
\begin{proof}
For $A\subseteq\{1,\hdots,n^2-1\}$ let $N(A):=\{X\in \mathcal{B}(\mathbb{C}^n):\text{tr}(X)=0\land\,\forall i\in S: \text{tr}(Xe_i)=0\}$ and note that $\dim_\mathbb{C} N(A)= n^2-1-|A|$ since the set of operators $\{\id,e_1,\hdots,e_{n^2-1}\}$ is linearly independent over $\mathbb{C}$. From this it follows that for all $A\subseteq\{1,\hdots,n^2-1\}$ we have $\dim (N(A)\cap H(n)_0)\leq n^2-1-|A|$: Assume for a contradiction that $\dim (N(A)\cap H(n)_0)=m> n^2-1-|A|$. Then there are is a set $\{h_1,\hdots,h_m\}\subseteq H(n)_0\subseteq N(A)$ of linearly independent operators over $\mathbb{R}$. Being a set of hermitian operators, $\{h_1,\hdots,h_m\}$ is also linearly independent over $\mathbb{C}$ and hence we conclude that $\dim N(A)\geq m$, a contradiction. Now note that $D_k=\left(\bigcup_{A\subseteq\{1,\hdots,n^2-1\}:|A|=n^2-1-k}N(A)\right)\cap S_{H(n)_0}$. Thus $D_k$ is a real algebraic set and $\dim D_k\le \left(n^2-1-(n^2-1-k)\right)-1=k-1$.

\end{proof}
For $k\in\mathbb{N}$ let $Q_k:=D_k-D_{k-1}$. For $k\in\{1,\hdots,l-1\}$ let $\mathcal{D}_k:=\mathcal{D}\cap Q_k$ and note that $\mathcal{D}_k$ is a semi-algebraic set with $\dim \mathcal{D}_k\leq k-1$. Furthermore let  $\mathcal{D}_l:=\mathcal{D}-\bigcup_{k=1}^{l-1} \mathcal{D}_k$ and note that $\dim\mathcal{D}_l\le \dim\mathcal{D}$. Then we have $\bigcup_{k=1}^{l}\mathcal{D}_k=\mathcal{D}$.

Just like in Subsection \ref{proof1}, we identify the set of $m$-dimensional POVMs on $\mathbb{C}^n$ with the semi-algebraic subset $\eta(\mathcal{P}(m))$ of $(H(n))^m$. Let $K:=(H(n))^m$. Just like in the proof of Theorem \ref{thmUfull}, for $j\in\{0,\hdots,l-1\},\ i\in\{1,\hdots,m\}$, define real polynomials
\begin{align*}
p_{i,j}:&K\times H_0(n)\simeq \mathbb{R}^{mn^2}\times\mathbb{R}^{n^2-1}\to\mathbb{R},\\
&(P,X)\mapsto\text{tr}\left(Q_i(\mathcal{T})^j(X)\right).
\end{align*}
Denote by $\mathcal{V}$ the common zero locus of the polynomials $\{p_{i,j}\}_{j\in\{0,\hdots,l-1\},\ i\in\{1,\hdots,m\}}$
and set $\mathcal{M}_k:=(K\times \mathcal{D}_k)\cap\mathcal{V}$ for $k\in\{1,\hdots,l\}$.

Next we prove that for all $k\in\{1,\hdots,l\}$ we have $\dim\mathcal{M}_k<\dim K$. So let $k\in\{1,\hdots,l\}$. If $\mathcal{D}_k=\emptyset$ we have $\mathcal{M}_k=\emptyset$ and thus we clearly have $\dim\mathcal{M}_k<\dim K$. Otherwise let $X\in \mathcal{D}_k$ be fixed. Using feasibility of the map $\mathcal{T}$, it is seen, by going along the lines of the proof of Proposition \ref{propind}, that the set of operators $\{(\mathcal{T})^i(X)\}_{i\in\{0,\hdots,k-1\}}$ is linearly independent over $\mathbb{C}$. This implies that at least $m\cdot k$ of the linear equations
\begin{align}\label{eqm}
\text{tr}\left(Q_i\mathcal{T}^j(X)\right)=0,\ j\in\{0,\hdots,l-1\},\ i\in\{1,\hdots,m\},\ (Q_1,\hdots,Q_m)\in K,
\end{align}
are independent. Hence the dimension of the solution set of the equations \eqref{eqm} is by at least $m\cdot k$ smaller than the dimension of $K$. Since this holds for all $X\in\mathcal{D}^{1}_k$  we find
\begin{align*}
\dim \mathcal{M}_k&\leq \dim K+\dim\mathcal{D}_k-km\leq \dim K+(k-1)-km\\
&\le \dim K -(m-1)k-1<\dim K
\end{align*}
for $k<l$. For $k=l$ we find 
\begin{align*}
\dim \mathcal{M}_k&\leq \dim K+\dim\mathcal{D}_l-lm\leq \dim K+\dim\mathcal{D}-lm\\
&< \dim K+lm-lm<\dim K,
\end{align*}
where we used the assumption that $lm>\dim\mathcal{D}$.

Let $\pi_1:K\times H(n)_0\to K$ be the projection on the first factor $K$ and let $\mathcal{M}:=\bigcup_{k=1}^{l}\mathcal{M}_k$. Clearly, $\mathcal{M}$ is a semi-algebraic set. Let \begin{align*}
K_{\mathcal{R}}:=\{P\in \mathcal{P}(m): (\mathcal{T^\dagger})^l(P)\ \text{ is not }\mathcal{R}-\text{complete.}\}.
\end{align*}
Then, since $\eta(\mathcal{P}(m))\subseteq K$ and $\mathcal{D}$ represents $\Delta(\mathcal{R})$, we have $\eta(K_{\mathcal{R}})\subseteq\pi_1(\mathcal{M})$. We have shown that $\dim \mathcal{M}_k<\dim K$ for all $k\in\{1,\hdots,l\}$ and thus $\dim \mathcal{M}<\dim K$. Hence we find $\dim\pi_1(\mathcal{M})<\dim K$ by Theorem 2.8.8 of \cite{bochnak1998real}. But since $\eta(K_{\mathcal{R}})$ is a subset of $\pi_1(\mathcal{M})$, this implies that $\eta(K_{\mathcal{R}})$ has measure zero in $K$.

Finally, since $\mathcal{D}$ is a closed subset of $S_{H(n)_0}$, stability follows form Lemma IV.1 of \cite{kech2}.

\appendix
\section{Continuous Time Evolution}\label{B}
In this appendix we consider continuous dynamics generated by Lie semigroups. Let $\mathcal{L}:\mathcal{B}(\mathbb{C}^n)\to\mathcal{B}(\mathbb{C}^n)$ be a unital conditional completely positive map generating the one parameter family of unital CP maps  $\mathcal{T}_t:=e^{t\mathcal{L}}:\mathcal{B}(\mathbb{C}^n)\to\mathcal{B}(\mathbb{C}^n)$ where $t\in\mathbb{R}_0^+$.

Instead of measuring the initial POVM after equidistant time steps of the system dynamics given by $\mathcal{T}_t$ we now consider more general time steps. More precisely, for $T:=(t_1,\hdots,t_l)$ a tuple of rational numbers such that $0<t_1<t_2<\hdots<t_l<1$ and $P:=(Q_1,\hdots,Q_m)$ a POVM define the measurement-scheme 
\begin{align*}
\mathcal{L}_T(P):=\left((P_1,\hdots,P_m),(\mathcal{T}_{t_1}(P_1),\hdots,\mathcal{T}_{t_1}(P_m)),\hdots,(\mathcal{T}_{t_l}(P_1),\hdots,\mathcal{T}_{t_l}(P_m))\right).
\end{align*}

To obtain generalizations of Theorem \ref{thmUprior} and Theorem \ref{thmUCPTP} it suffices to generalize Proposition \ref{propind} to rational points in time, the remainder of their proofs can be transferred. For $T:=(t_1,\hdots,t_l)\in\mathbb{Q}^l$ such that $0<t_1<t_2<\hdots<t_l<1$ and $\Lambda:=\{\lambda_1,\hdots,\lambda_l\}\subseteq \mathbb{C}-\{0,1\}$ a subset define
\begin{align*}
V_T^\Lambda:=\begin{pmatrix}
1 & 1 & \dots & 1 \\
1 & \lambda_{1}^{t_1} & \dots & \lambda_{1}^{t_l} \\
\vdots & \vdots &  & \vdots \\
1 & \lambda_{l}^{t_1} & \dots& \lambda_{l}^{t_l}
\end{pmatrix}.
\end{align*}
\begin{proposition}\label{propgenind}
$V_T^\Lambda$ is invertible.
\end{proposition}
\begin{proof}
Let $N,m_1,\hdots,m_l\in \mathbb{N}$ be such that $t_i=\frac{m_i}{N}$. Extend $\Lambda$ to a set $\tilde{\Lambda}:=\{\lambda_1,\hdots,\lambda_{m_l}\}\subseteq\mathbb{C}-\{0,1\}$ . Let $\tilde{\lambda}_i:=\lambda_i^{\frac1{N}},\ i\in\{1,\hdots,m_l\}$. Then $\det V_T^\Lambda$ is a minor of the matrix
\begin{align*}
V:=\begin{pmatrix}
1 & 1 & \dots & 1 & \dots & 1 \\
1 & \tilde{\lambda}_{1} & \dots & \tilde{\lambda}_{1}^{m_1}& \dots & \tilde{\lambda}_{1}^{m_l} \\
\vdots & \vdots &  & \vdots & & \vdots\\
1 & \tilde{\lambda}_{m_1} & \dots & \tilde{\lambda}_{m_1}^{m_1}& \dots & \tilde{\lambda}_{m_1}^{m_l} \\
\vdots & \vdots &  & \vdots & & \vdots\\
1 & \tilde{\lambda}_{m_l} & \dots & \tilde{\lambda}_{m_l}^{m_1}& \dots & \tilde{\lambda}_{m_l}^{m_l}
\end{pmatrix}.
\end{align*}
Multiplying $V$ from the left with the diagonal matrix $\text{Diag}(\tilde{\lambda}_1,\hdots,\tilde{\lambda}_{m_l})$ gives a Vandermonde matrix. This Vandermonde matrix is totally non-singular since $\tilde{\Lambda}\subseteq\mathbb{C}-\{0,1\}$ by construction. Thus $\det V_T^\Lambda$ does not vanish. 
\end{proof}

From Proposition \ref{propgenind} we directly obtain the following generalizations of Theorem \ref{thmUprior} and Theorem \ref{thmUCPTP} respectively. 

\begin{corollary}
Let $T:=(t_1,\hdots,t_l)\in\mathbb{Q}^l$ such that $0<t_1<\hdots<t_l<1$. Let $h\in H(n)$ be such that $e^{ih}$ is feasible. For $\mathcal{R}\subseteq \mathcal{S}(\mathbb{C}^n)$ a subset, let $\mathcal{D}$ be a semi-algebraic set that represents $\Delta(\mathcal{R})$. Let $m\geq n$ and let $l\in\mathbb{N}$ be such that $(l+1)(m-1)>\dim \mathcal{D}$, then for almost all POVMs $P$ with $m$ outcomes, the measurement scheme $\mathcal{T}_U^l(P)$ is stably $\mathcal{R}$-complete.
\end{corollary}
\begin{corollary}
Let $T:=(t_1,\hdots,t_l)\in\mathbb{Q}^l$ such that $0<t_1<\hdots<t_l<1$. Let $\mathcal{L}$ be a unital conditional completely positive map such that $e^\mathcal{L}$ is feasible. For $\mathcal{R}\subseteq\mathcal{S}(\mathbb{C}^n)$ a subset, let $\mathcal{D}$ be a semi-algebraic set that represents $\Delta(\mathcal{R})$. Let $m\in\mathbb{N}$ and let $l\in\mathbb{C}$ be such that $(l+1)(m-1)>\dim\mathcal{D}$, then for almost all POVMs $P$ with $m$ outcomes, the measurement scheme $\mathcal{T}^l(P)$ is stably $\mathcal{R}$-complete.
\end{corollary}

\bibliographystyle{unsrt}
\bibliography{bibliography}

\end{document}